\documentclass[conference,a4paper]{IEEEtran}

\usepackage{url}
\usepackage{ifthen}
\usepackage{graphicx}
\usepackage{mathtools}
\usepackage{etoolbox}
\usepackage{lipsum}
\usepackage{amsfonts}
\usepackage{cite}
\usepackage{amsmath,amsthm}
\usepackage{amssymb}
\PassOptionsToPackage{bookmarks={false}}{hyperref}
\let\bbordermatrix\bordermatrix
\patchcmd{\bbordermatrix}{8.75}{4.75}{}{}
\patchcmd{\bbordermatrix}{\left(}{\left[}{}{}
\patchcmd{\bbordermatrix}{\right)}{\right]}{}{}

\newtheorem{theorem}{Theorem}[section]

\newtheorem{definition}[theorem]{Definition}

\newtheorem{proposition}[theorem]{Proposition}

\newcommand{\sr}{\stackrel}

\newcommand{\rar}{\rightarrow}

\newcommand{\tri}{\sr{\triangle}{=}}

\newcommand{\be}{\begin{equation}}
\newcommand{\ee}{\end{equation}}
\newcommand{\bea}{\begin{eqnarray}}
\newcommand{\eea}{\end{eqnarray}}
\newcommand{\bes}{\begin{eqnarray*}}
\newcommand{\ees}{\end{eqnarray*}}
\newcommand{\beae}{\begin{IEEEeqnarray}{rCl}}
\newcommand{\eeae}{\end{IEEEeqnarray}}

\newcommand{\bi}{\begin{itemize}}
\newcommand{\ei}{\end{itemize}}
\newcommand{\ben}{\begin{enumerate}}
\newcommand{\een}{\end{enumerate}}
\newcommand{\bp}{\begin{problem}}
\newcommand{\ep}{\end{problem}}

\newcommand{\hst}{\hspace{.2in}}

\begin{document}

\title{Source-Channel Matching for Sources with Memory}

\author{
  \IEEEauthorblockN{Christos Kourtellaris, Photios A.~Stavrou, Charalambos D.~Charalambous}
  \IEEEauthorblockA{Dep. of Electrical \& Computer Engineering, University of Cyprus, Nicosia, Cyprus\\
    Email: \{kourtellaris.christos, stavrou.fotios, chadcha\}@ucy.ac.cy}

}



\maketitle

\begin{abstract}
To be considered for an IEEE Jack Keil Wolf ISIT Student Paper Award. In this paper we  analyze the probabilistic matching of
sources with memory to  channels with memory so that  symbol-by-symbol code with memory without anticipation
are optimal, with respect to an average distortion and excess distortion
probability. We show achievability of such a symbol-by-symbol code with memory without anticipation, and we show matching for   the Binary Symmetric Markov source (BSMS(p)) over a first-order symmetric channel with a cost constraint.
\end{abstract}
\section{Introduction}
\par In this paper we address the problem of Joint Source-Channel Coding  JSCC based  on symbol-by-symbol code transmission with memory without anticipation. Thus, at each instant of time $i$, we   impose  real-time transmission constrains on  the encoder
and decoder to process samples independently, with memory on past symbols, and without anticipation with respect to symbols occurring future times $j>i$. The aim is to
match probabilistically the source to a channel, and evaluate its performance with respect to excess distortion probability.
\par For memoryless sources and channels, necessary
and sufficient conditions for symbol-by-symbol transmission are given in \cite{gastpar2003} (see also \cite{kverdu}). However, extending these results to sources with  memory is not a trivial task for the following two reasons. i) The optimal reproduction distribution of classical Rate Distortion Function (RDF), used during the realization procedure, to match the source to a channel is, in general noncausal (anticipative on future symbols); ii) the solution to the RDF is often unknown.
\par In this paper we consider a nonanticipative information RDF which
is realizable in the above sense, and we proceed to obtain the  expression of  the optimal causal
reproduction distribution.
1) We prove under certain conditions involving the nonanticipative information RDF, and the capacity
of certain channels with memory and feedback, that symbol-by-symbol code with memory without anticipation is achievable.
2) we consider a BSMS(p) and we show that matching  is possible over a symmetric channel with memory and cost constraint, 3) we evaluate the excess distortion probability and we show that convergence to zero, as the  number of channel uses increases, establishing achievability.
\section{Symbol-by-Symbol codes with Memory Without Anticipation}
\label{sbs}
\par Let ${\mathbb{N}}\tri\{0,1,\dots\}$, $\mathbb{N}^n\tri\{0,1,\dots,n\}$. The spaces ${\cal X},{\cal A},{\cal B},{\cal Y}$ denote the source output, channel input,
channel output, and decoder output alphabets, respectively, which are assumed to be complete separable metric spaces (Polish spaces)  to avoid excluding continuous alphabets. We define their product spaces by  ${\cal X}_{0,n}\tri\times_{i=0}^{n}{\cal X}$,
${\cal A}_{0,n}\tri\times_{i=0}^{n}{\cal A}$, ${\cal B}_{0,n}\tri\times_{i=0}^{n}{\cal B}$,
${\cal Y}_{0,n}\tri\times_{i=0}^{n}{\cal Y}$, and associate them with their measurable spaces. Let
$x^n\tri\{x_0, x_1,\dots, x^n\}\in{\cal X}_{0,n}$ denote the source sequence
of length $n+1$, and similarly for the rest of the blocks. 
 Next, we introduce the various distributions..

\begin{definition}
\label{source}
(Source) The source is a sequence of conditional
distributions 
defined by
\vspace{-0.3cm}
\begin{align}
P_{X^n}(d{x}^n)\tri\otimes_{i=0}^{n} P_{X_i|X^{i-1}}(d{x}_i|x^{i-1}).\nonumber
\end{align}
\end{definition}
\vspace{-0.2cm}

\begin{figure}
\begin{center}
\includegraphics[bb= -10 23 400 200, scale=0.5]{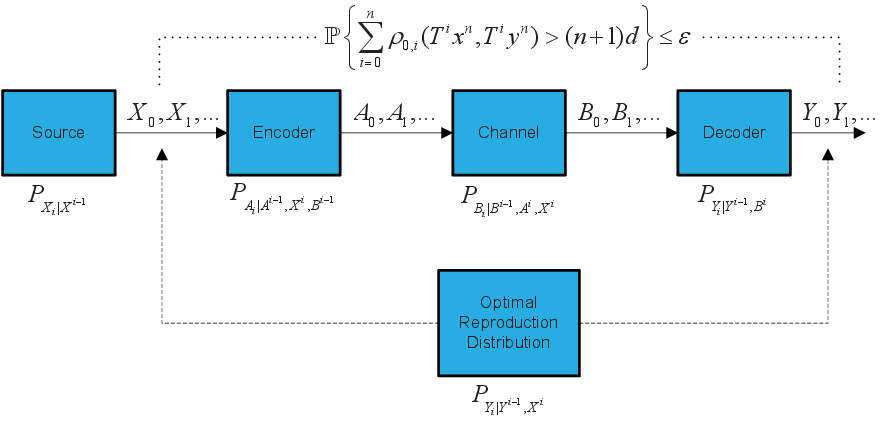}
\caption{Communication scheme with feedback.}
 \label{cs}
\end{center}
\end{figure}

\begin{definition}
\label{encoder}
(Encoder) The encoder is a sequence of conditional
distributions 
defined by
\begin{align}
&{\overrightarrow P}_{A^n|B^{n-1},X^n}(d{a}^n|b^{n-1},x^n)\nonumber\\
&\tri\otimes_{i=0}^{n}P_{A_i|A^{i-1},B^{i-1},X^i}(d{a}_i|a^{i-1},b^{i-1},x^i). \nonumber
\end{align}
\end{definition}
\vspace{-0.1cm}
Thus, the encoder is nonanticipative in the sense that at each time $i\in{\mathbb{N}}^n$, $P_{A_i|A^{i-1},B^{i-1},X^i}(d{a}_i|a^{i-1},b^{i-1},x^i)$ is a measurable function of past and present symbols $x^i\in{\cal X}_{0,i}$ and past symbols $a^{i-1}\in{\cal A}_{0,i-1}, b^{i-1}\in{\cal B}_{0,i-1}$.
\begin{definition}
\label{channel}
(Channel) The channel is a sequence of conditional
distributions 
defined by
\begin{align}
&{\overrightarrow P}_{B^n|A^{n},X^n}(d{b}^n|a^{n},x^n)\nonumber\\
&\tri\otimes_{i=0}^{n}P_{B_i|B^{i-1},A^{i},X^i}(d{b}_i|b^{i-1},a^{i},x^i).\nonumber
\end{align}
\end{definition}

Thus the channel has memory, feedback and it is nonanticipative with respect to  the source sequence.

\begin{definition}
\label{decoder}
(Decoder) The decoder is a sequence of conditional
distributions 
defined by
\begin{align}
{\overrightarrow P}_{Y^n|B^{n}}(d{y}^n|b^{n})\tri
\otimes_{i=0}^{n}P_{Y_i|Y^{i-1},B^i}(d{y}_i|y^{i-1},b^i).\nonumber
\end{align}
\end{definition}

Definitions~\ref{source}-\ref{decoder} are  general, since they allow memory and feedback  without anticipation, hence we call the source-channel code symbol-by-symbol code with memory without anticipation.
Given the source, encoder, channel,  decoder, we can define uniquely
the joint measure  by
\begin{align}
& P_{X^n,A^n,B^n,Y^n}(d{x}^n,d{a}^n,d{b}^n,d{y}^n)\nonumber\\
& =\otimes_{i=0}^{n}P_{Y_i|Y^{i-1},B^i}(d{y}_i|y^{i-1},b^i)\nonumber\\
&\otimes P_{B_i|B^{i-1},A^{i},X^i}(d{b}_i|b^{i-1},a^{i},x^i)\nonumber\\
&  \otimes P_{A_i|A^{i-1},B^{i-1},X^i}(d{a}_i|a^{i-1},b^{i-1},x^i)\otimes P_{X_i|X^{i-1}}(d{x}_i|x^{i-1}).  \label{joint}
\end{align}
The previous equation  implies the  Markov Chains (MCs):
\begin{align}
& (A^{i-1},B^{i-1},Y^{i-1})\leftrightarrow X^{i-1}\leftrightarrow X_i, \ \ \forall i\in\mathbb{N}^n \label{mc1} \\
& Y^{i-1}\leftrightarrow (A^{i-1},B^{i-1},X^{i})\leftrightarrow A_i, \ \ \forall i\in\mathbb{N}^n \label{mc2}\\
& Y^{i-1}\leftrightarrow (A^{i},B^{i-1},X^{i})\leftrightarrow B_i, \ \ \forall i\in\mathbb{N}^n \label{mc3}\\
& (A^{i},X^{i})\leftrightarrow (B^i, Y^{i-1})\leftrightarrow Y_i, \ \ \forall i\in\mathbb{N}^n. \label{mc4}
\end{align}
\par The distortion between the source and its reproduction is a measurable function $d_{0,n}:{\cal X}_{0,n}\times{\cal Y}_{0,n}\mapsto [0,\infty)$,
and the cost of transmitting  symbols over
the channel is a measurable function $c_{0,n}: {\cal A}_{0,n} \times {\cal Y}_{0,n-1} \mapsto [0,\infty)$ defined by
\begin{align}
d_{0,n}(x^n,y^n)\tri &\sum_{i=0}^{n}{\rho}_{0,i}({T}^i{x^n},T^i{y^n})\nonumber\\
c_{0,n}(a^n,b^{n-1}) \tri &\sum_{i=0}^{n}
{\gamma}_{0,i}(a^i,b^{i-1}),\nonumber
\end{align}
where $({T}^i{x^n},{T}^i{y^n})$ are the shift operations on $(x^n,y^n)$, respectively.
For a single letter distortion function we take ${\rho}_{0,i}(T^ix^n,T^iy^n)={\rho}(x_i,y_i)$.
 Next, we state the definition of a symbol-by-symbol code (with memory without anticipation).

\begin{definition}
\label{sbsc}
(Symbol-by-Symbol Code)
 An (n,d,$\epsilon$,P) symbol-by-symbol
code for (${\cal X}_{0,n}, {\cal A}_{0,n}, {\cal B}_{0,n}, {\cal Y}_{0,n}, P_{X^n},
{\overrightarrow P}_{B^n|A^n,X^n}, d_{0,n}, c_{0,n}$) is a  code
$\{P_{A_i|A^{i-1}, B^{i-1},X^i}(\cdot|\cdot):\forall i \in\mathbb{N}^n\}$, $\{P_{Y_i|Y^{i-1},B^i}(\cdot|\cdot):\forall i \in\mathbb{N}^n\}$
with excess distortion probability
${\mathbb P}\Big\{d_{0,n}(x^n,y^n)>(n+1)d\Big\}\leq\epsilon, \ \epsilon\in(0,1), \ d\geq 0,$
and transmission cost $\frac{1}{n+1} {\mathbb E}\Big\{c_{0,n}(A^n,B^{n-1})\Big\}\leq P, \  P\geq 0$.
\end{definition}

\begin{definition}(Minimum Excess Distortion)
\label{asbsc}
The minimum excess distortion achievable by a symbol-by-symbol
code $(n,d,\epsilon,P)$ is defined by
\begin{align}
&  D^o(n,\epsilon, P)\tri\inf\Big\{d:  \exists(n,d,\epsilon,P)\  \ \mbox{symbol-by- symbol code}\Big\} \nonumber
\end{align}
\end{definition}

Our definition of symbol-by-symbol code is randomized, hence it embeds  deterministic codes as a special case.

\section{Nonanticipative RDF}
\label{cnrdf}

\par The necessary conditions for transmitting
a symbol-by-symbol code (they also hold  for  memoryless sources and channels) is the following.

\begin{enumerate}
\item Computation of the RDF and that of the optimal reproduction distribution so that  probabilistic matching of the source and  channel is feasible;

\item Realization
of the optimal reproduction distribution of lossy compression with fidelity by an encoder-channel-decoder scheme, processing information causally.
\end{enumerate}
\par Therefore, to facilitate the matching we introduce the RDF.
%
 Given a source distribution $P_{X^n}(\cdot)$ and a reproduction distribution
$P_{Y^n|X^n}(\cdot|x^n)$ the average fidelity set is
\begin{align}
&{\cal Q}_{0,n}(D)\tri \Big\{{P}_{Y^n|X^n}: \nonumber  \\
&\frac{1}{n+1}\int d_{0,n}(x^n,y^n)
({P}_{Y^n|X^n}\otimes P_{X^n})(dx^n,dy^n)\leq D\Big\}.  \nonumber 
\end{align}
It is known that for stationary ergodic sources, the OPTA   is given by the RDF \cite{berger}
$R(D)=\lim_{n\rar \infty} R_{0,n}(D)$,
$R_{0,n}(D)=\inf_{{P}_{Y^n|X^n}\in {\cal Q}_{0,n}(D)}\frac{1}{n+1}I(X^n;Y^n)$,
provided the infimum is achievable. However, $R(D)$ is only known for IID
and Gaussian sources, and in generally fails to satisfy 1), 2).

Now, we  introduce the nonanticipative information RDF which by construction is realizable.
 Given a source ${P}_{X^n}(dx^n)$ and a causal conditional distribution defined by
\begin{align}
\overrightarrow{P}_{Y^n|X^n}(dy^n|x^n) \tri \otimes_{i=0}^{n}P_{Y_i|Y^{i-1},X^i}(dy_i|y^{i-1},x^i) \label{cc1}
\end{align}
we introduce the information measure
\begin{align}
I_{P_{X^n}}(X^n\rar Y^n)&\tri\mathbb{D}({\overrightarrow P}_{Y^n|X^n}\otimes P_{X^n}||{P}_{Y^n}\times P_{X^n})\nonumber\\
&\equiv\mathbb{I}_{X^n\rightarrow{Y^n}}(P_{X^n},{\overrightarrow P}_{Y^n|X^n}).\nonumber
\end{align}
Consider the fidelity set defined by
\begin{align}
\overrightarrow{\cal Q}_{0,n}(D) &\tri  \Big\{{\overrightarrow P}_{Y^n|X^n}: \frac{1}{n+1}\int_{{\cal X}_{0,n}\times {\cal Y}_{0,n}}d_{0,n}({x^n}
,{y^n})\nonumber\\
&\qquad {\overrightarrow P}_{Y^n|X^n}(dy^n|x^n)\otimes P_{X^n}(dx^n) \leq D\Big\}.\label{nafs}
\end{align}

\begin{definition}(Nonanticipative Information RDF)
\label{nardf}
Given ${\overrightarrow {\cal Q}}_{0,n}(D)$, the nonanticipative information RDF is defined by
\begin{align}
{ R}^{na}_{0,n}(D) \tri \inf_{ \overrightarrow{P}_{Y^n|X^n}\in   {\overrightarrow {\cal Q}}_{0,n}(D)} \frac{1}{n+1} \mathbb{I}_{X^n\rightarrow{Y^n}}(P_{X^n},{\overrightarrow P}_{Y^n|X^n})\label{nardf11}
\end{align}
and its rate by ${ R}^{na}(D)=\lim_{n\rar\infty}{ R}^{na}_{0,n}(D)$
provided infimum and the limit exist.
\end{definition}

Clearly, if the minimum of ${ R}^{na}_{0,n}(D)$ exists the  optimal reproduction distribution
is nonanticipative, and hence realizable.\\
It can be shown that $R_{0,n}^{na}(D)$ is equal to  the nonanticipatory $\epsilon-$entropy introduced by Gorbunov and Pinsker in   \cite{pinsker1973}, via
\begin{align}
R^{\varepsilon}_{0,n}(D)=\mathop{\inf_{{P}_{Y^n|X^n}\in {\cal Q}_{0,n}(D)}}_{{X_{i+1}^{n}}\leftrightarrow X^i \leftrightarrow Y^i, \ i=0,1,\dots,n-1
}\frac{1}{n+1}I(X^n;Y^n) \label{rdfon}
\end{align}



The MC in (\ref{rdfon}) implies that  the reproduction distribution which minimizes
(\ref{rdfon}) can be realized via an encoder-channel-decoder, using nonanticipative operations (causal).




\label{ord}


Under the conditions in  \cite{pinsker1973}, or assuming the solution of $R_{0,n}^{na}(D)$  is stationary, which implies ${\overrightarrow P}_{Y^n|X^n}(d{y}^n|x^n)$ is a stationary conditional distribution, we have the following theorem \cite{pscc}.

\begin{theorem}
\label{maintheo}
Suppose there exist an interior point of the fidelity set, and the optimal reproduction is stationary.  Then the infimum over $\overrightarrow{\cal Q}_{0,n}(D)$ in (\ref{nardf11}) is
attained by
\begin{align}
\overrightarrow{P}_{Y^n|X^n}^*(dy^n|x^n)=\otimes_{i=0}^{n}\frac{e^{s\rho(T^i{x}^n,T^i{ y}^{n})}
P_{Y_i|Y^{i-1}}^*(d{y}_{i}|{ y}^{i-1})}{\int_{{\cal  Y}_i}
e^{s\rho(T^i{x}^n,T^i{ y}^{n})}P_{Y_i|Y^{i-1}}^*(d{ y}_{i}|{ y}^{i-1})}  \label{crdoo}
\end{align}
where $s\leq 0$ is the Lagrange multiplier associated with the constraint which is satisfied with equality, and
\begin{align}
{ R}^{na}_{0,n}(D)=& sD - \frac{1}{n+1} \sum_{i=0}^{n}\int_{{\cal X}_{0,i}\times{{\cal Y}}_{0,i-1}}
\log\Big(\int_{{{\cal Y}}_{i}}e^{s\rho(T^i{x}^n, T^i{ y}^{n})}\nonumber\\
& P_{Y_i|Y^{i-1}}^{*}(d{y}_{i}|{ y}^{i-1})\Big) \otimes {P}_{X_i|X^{i-1}}(d{x}_i|{ x}^{i-1})\nonumber \\
&\otimes P_{X^{i-1},Y^{i-1}}^*(d{x}^{i-1},d{y}^{i-1}){\label{solrdf}}
\end{align}
where $P_{X^{i-1},Y^{i-1}}^*(\cdot,\cdot)= \overrightarrow{P}_{Y^{i-1}|X^{i-1}}^*(\cdot|\cdot)\otimes P_{X^{i-1}}(\cdot)$.
\end{theorem}
\begin{proof}  The derivation is given in \cite{charalambous-stavrou-ahmed2013}.
\end{proof}
Clearly,  (\ref{crdoo})  is nonanticipative, and  as we show in the next section, easy to compute, even for sources with memory.
\section{Coding Theorem}
\label{coding}
 In this section we show achievability of symbol-by-symbol code. First, we define the probabilistic realization of optimal reproduction distribution.


\begin{definition}\label{realdef}
(Realization)
Given a source $\{P_{X_i|X^{i-1}}$ $(d{x}_i|x^{i-1}): \forall i \in {\mathbb N}^n\}$, a general channel
$\{P_{B_i|B^{i-1},A^i,X^i}$ $(d{b}_i|b^{i-1},a^i,x^i): \forall i \in {\mathbb N}^n\}$  is a realization
of the optimal reproduction distribution $\{P_{Y_i|Y^{i-1},X^i}^*(d{y}_i|y^{i-1},x^i): \forall i \in {\mathbb N}^n\}$
of theorem \ref{maintheo}, if there exists a pre-channel encoder
$\{P_{A_i|A^{i-1},B^{i-1},X^i}$ $(d{a}_i|a^{i-1},b^{i-1},x^i): \forall i \in {\mathbb N}^n\}$ and a post-channel
decoder $\{P_{Y_i|Y^{i-1},B^{i}}$ $(d{y}_i|y^{i-1},b^{i}): \forall i \in {\mathbb N}^n\}$ such that
\begin{align}
{\overrightarrow P}^*_{Y^n|X^n}(d{y}^n|x^n)&=\otimes_{i=0}^n{ P}^*_{Y_i|Y^{i-1},X^i}
(d{y}_i|y^{i-1},x^i)\nonumber\\
&=\otimes_{i=0}^n{ P}_{Y_i|Y^{i-1},X^i}
(d{y}_i|y^{i-1},x^i)\label{scmrd}
\end{align}
where the joint distribution from which (\ref{scmrd}) is obtained  is given precisely by (\ref{joint}).
Moreover we say that ${ R}^{na}_{0,n}(D)$ is realizable if in addition the realization
operates with average distortion $D$ and $I_{P_{X^n}}(P_{X^n},\overrightarrow{P}_{Y^n|X^n})={ R}^{na}_{0,n}(D)$
\end{definition}

If the optimal reproduction distribution is realizable (see  Definition~\ref{realdef}), then  the data processing inequality holds:
\begin{align}
I_{X^n\rar Y^n}(P_{X^n},{\overrightarrow P}_{Y^n|X^n})\leq I(X^n\rar B^n), \ \forall n \in{\mathbb{N}}. \label{dpi}
\end{align}
If ${ R}^{na}_{0,n}(D)$ is realizable according to Definition \ref{realdef},
then the source is not necessarily matched to the channel. Next, we  prove (under certain conditions) achievability.\\
Consider the following average cost set defined by
\begin{align}
{\cal P}_{0,n}(P)\tri\Big\{(X^n,A^n):\frac{1}{n+1}{\mathbb E}\{c_{0,n}(A^n,B^{n-1})\}\leq P\Big\}. \nonumber
\end{align}
Since we consider the general scenario that (\ref{mc1})-(\ref{mc4}) hold, then we define the information channel capacity as follows \cite{cover-pombra1989}.
\begin{align}
C_{0,n}(P)\tri\sup_{(X^n,A^n)\in{\cal P}_{0,n}(P)}\frac{1}{n+1}I(X^n\rar B^n)\nonumber
\end{align}
and its rate (provided $\sup$ is finite and the limit exists) by
$C(P)=\lim_{n\rar\infty}C_{0,n}(P)$.

Next, we prove achievability of a symbol-by-symbol code.

\begin{theorem}\label{ach}
(Achievability of Symbol-by-Symbol Code). Suppose the following conditions hold.

{\bf (1)} ${ R}^{na}_{0,n}(D)$ has a solution, and the optimal reproduction distribution is stationary of the form $\{P_{Y_i| Y^{i-1}, X^i}: \forall i=0, 1, \ldots, n\}$;

{\bf (2)}  $C_{0,n}(P)$ has a solution, the maximizing processes are stationary, and the encoder is of the form $\{P_{A_i| A^{i-1}, X^i}: \forall i=0, 1, \ldots, n\}$;

{\bf (3) } The optimal reproduction distribution $\overrightarrow{P}_{Y^n|X^n}(dy^n|x^n)$ given by Theorem \ref{maintheo} is realizable, and ${ R}^{na}_{0,n}(D)$ is also realizable.

{\bf (4)} For a given $D$ there exists a $P$ such that ${R}^{na}(D)=C(P)$.
\begin{align}
\mbox{If} \hst \mathbb{P}\Big\{\sum_{i=0}^{n}{\rho}_{0,i}(T^i{X^n},T^i{Y^n})>(n+1)d\Big\}\leq\epsilon \label{edp}
\end{align}
where ${\mathbb P}$ is taken with respect to $P_{Y^n,X^n}(d{y}^n,d{x}^n)={\overrightarrow P}^*_{Y^n|X^n}(d{y}^n|x^n)\otimes P_{X^n}(d{x}^n)$
induced by matching, then there exists an $(n,d,\epsilon,P)$ symbol-by-symbol code with memory without anticipation.
\end{theorem}

\begin{proof}
The derivation is similar to \cite{gastpar2003}. If conditions (1), (3) hold then the optimal reproduction distribution is realizable, and this realization achieves ${ R}^{na}_{0,n}(D)$. By (4) the source is matched to the channel so that the excess distortion probability of a symbol-by-symbol code  with memory without anticipation satisfies (\ref{edp}).
\end{proof}

\subsection{\bf Existence of Symbol-by-Symbol Codes} Next, we give sufficient conditions so that the conditions of Theorem~\ref{ach}, {\bf (1), (2)} hold, i.e., establishing  existence of a symbol-by symbol encoder $\{P_{A_i| A^{i-1}, B^{i-1}, X^i}: i=0, 1, \ldots, n\}$. Suppose the following conditions hold.
\begin{itemize}
\item[\bf (A1)] $\rho_{0,i}(T^ix^n, T^i y^n)= \rho_{0,i}(x_i, T^i y^n), \forall i \in {\mathbb N}^n$;

\item[\bf (A2)]  $P_{X_i|X^{i-1}}(x_i|x^{i-1})=P_{X_i|X_{i-1}}(x_i|x_{i-1}), \ \forall i\in\mathbb{N}^n$;

\item[\bf (A3)] $P_{B_i|B^{i-1},A^{i},X^i}(d{b}_i|b^{i-1},a^{i},x^i) \\=P_{B_i|B^{i-1},A_{i},X_i}(d{b}_i|b^{i-1},a_{i},x_i), \ \forall i\in\mathbb{N}^n$.
\end{itemize}
If {\bf (A1)} holds, then by Theorem~\ref{maintheo} the optimal stationary reproduction distribution is $P_{Y_i| Y^{i-1}, X^i}^*=P_{Y_i| Y^{i-1}, X_i}^*, \forall i \in {\mathbb N}^n$, and  hence the form of the optimal reproduction distribution in  Theorem~\ref{ach}, {\bf (1)} holds. Moreover, if {\bf (A2), (A3)} hold,
then
 maximizing directed information $I(X^{n}\rightarrow  {B}^{n})$  over non-Markov encoders $\{P_{A_i|A^{i-1},B^{i-1},X^i}: i=0,1,\ldots,n\}$ is equivalent to maximizing it over encoders $\{\overline{P}_{A_i|B^{i-1},X_i}: i=0,1,\ldots,n\}$, and similarly,  maximizing $I(X^{n}\rightarrow  {B}^{n})$  over non-Markov deterministic encoders $\{e_i(x^i,a^{i-1},b^{i-1}): i=1, \ldots, n\}$ is equivalent to the maximization with respect to encoders $\{g_i(x_i,b^{i-1}): i=1, \ldots, n\}$. This result appeared in \cite{charalambous-kourtellaris-hadjicostis} and is calculated using dynamic programming. Hence, the form of the encoder in Theorem~\ref{ach}, {\bf (2)} holds. Thus, based on these two conditions the encoder is symbol-by-symbol Markov with respect to the source, and nothing can be gained by considering an encoder that depends on the entire past of the source causally.

\section{Symbol-by-Symbol JSCC of a Binary Symmetric Markov source
via a Binary State Symmetric Channel}\label{exa}
In this section we provide a noisy coding theorem for a Binary Symmetric Markov Source
with crossover probability $p$, $BSMS(p)$. This is achieved by symbol-by-symbol joint
source channel matching of the current source via a Binary State Symmetric Channel $BSSC(\alpha_1,\beta_1)$
with an average cost constraint. First, we give the expression of the nonanticipative
reproduction distribution which  achieves the infimum in (\ref{nardf11}).
Next, we give the capacity expression of the $BSSC(\alpha_1,\beta_1)$ and the
optimal input distributions without feedback that achieve it. For this channel
feedback does not increase the capacity. Then, by merging these results we show
achievability of symbol-by-symbol code such that, $R^{na}(D)=C(\kappa)$.
\subsection{Results on $BSMS(p)$ and $BSSC(\alpha_1,\beta_1)$}
Consider a Binary Symmetric Markov Source. $BSMS(p)$,
$P_{X_i|X_{i-1}}(0|0)=P_{X_i|X_{i-1}}(1|1)=1-p$ and $P_{X_i|X_{i-1}}(1|0)=P_{X_i|X_{i-1}}(0|1)=p$ and
$i\in{\mathbb N}^n$ and Hamming distortion criterion $\rho(x,y)=0$ if $x=y$ and $\rho(x,y)=1$ if $x \neq y$.
\begin{theorem}\label{marex1} For a BSMS(p) and single letter Hamming distortion criterion
${ R}^{na}(D)$ is given by
\[ { R}^{na}(D) = \left\{ \begin{array}{ll}
         H(p)-mH(\alpha)-(1-m)H(\beta) & \mbox{if $D \leq \frac{1}{2}$}\\
        0 & \mbox{otherwise}\end{array}  \right. \]
$m=1-p-D+2pD$, $\alpha=\frac{(1-p)(1-D)}{1-p-D+2pD}$, $\beta=\frac{p(1-D)}{p+D-2pD}$.
\end{theorem}
\begin{proof}
We describe the main steps. The steady state
distribution of the source is $P(X_i=0)=P(X_i=1)=0.5$ and the reproduction distribution is
\begin{align}
P_{Y_i|X^i,Y^{i-1}}^*=P_{Y_i|X_i,Y^{i-1}}^*=
\frac{e^{s\rho(x_i,y_i)}P(y_i|y^{i-1})}{\sum_{y_i}e^{s{\rho}(x_i,y_i)}P(y_i|y^{i-1})}\nonumber
\end{align}
and we can show that  $P_{Y_i|X_i,Y^{i-1}}^*=P_{Y_i|X_i,Y_{i-1}}^*$ and that
\begin{align}
P_{Y_i|X_i,Y_{i-1}}^*(y_i|x_i,y_{i-1})=\bbordermatrix{~ & 0,0 & 0,1 & 1,0 & 1,1 \cr
                  0 & \alpha & \beta& 1-\beta & 1-\alpha\vspace{0.3cm} \cr
                  1 & 1-\alpha & 1-\beta& \beta &  \alpha \cr}\nonumber
\end{align}
\par Using the stationary distributions $P_{Y_i|X_i,Y_{i-1}}^*$ and
$P_{X_i|X_{i-1}}$, we obtain $R^{na}(D)$.
\end{proof}
\par To perform the matching on the source to the channel we use the Binary State Symmetric Channel $BSSC(\alpha_1,\beta_1)$
defined by
\begin{IEEEeqnarray}{l}
 P_{B_i|A_i,B_{i{-}1}}(b_i|a_i,b_{i{-}1}) {=} \bbordermatrix{~ & 0,0 &\hspace{-0.15cm} 0,1 &\hspace{-0.15cm} 1,0 &\hspace{-0.15cm} 1,1 \cr
                  0 & \alpha_1 &\hspace{-0.15cm} \beta_1 &\hspace{-0.15cm} 1{-}\beta_1 &\hspace{-0.15cm} 1{-}\alpha_1 \cr
                  1 & 1{-}\alpha_1 &\hspace{-0.15cm} 1{-}\beta_1  &\hspace{-0.15cm} \beta_1 &\hspace{-0.15cm}  \alpha_1 \cr}.
                  \label{gench}\IEEEeqnarraynumspace
\end{IEEEeqnarray}
\par The form of the channel \ref{gench} is motivated by the form of the $P_{Y_i|X_i,Y_{i-1}}^*$
(as in the IID Bernoulli source is matched via a binary symmetric channel).
The state of the channel is defined as the modulo2 addition of the current
input and previous output symbol, $s_i=a_i\oplus b_{i-1}$. Then we may
transform the channel to its equivalent form defined by
$P_{B_i|A_i,S_{i}}(b_i|a_i,s_{i})$. This channel is called binary state symmetric channel,
since given the state the channel it is binary symmetric. We introduce a cost constraint on the channel
that has the following physical interpretation. Assume $\alpha_1>\beta_1\geq 0.5$. Then the capacity
of the state zero channel $(1-H(\alpha_1))$, is greater than
the capacity of the state one channel $(1-H(\beta_1))$. With``abuse" of terminology,
we interpret the $(BSC(1-\alpha_1))$ as the ``good channel" and the $(BSC(1-\beta_1))$
as the bad channel. It is further reasonable to assume that the we pay a larger
fee to use the ``good channel" and a smaller fee to use the ``bad channel". We quantify
this policy by assigning a binary pay off to each of the channels. Hence, we assign
a cost equal to $1$ for the good channel, and a cost equal to $0$ for the bad channel,
defined by
\[ c(a_i,b_{i-1}) \tri \left\{
  \begin{array}{l l}
    1 & \quad \text{if $a_i=b_{i-1}$, or $s_i=0$}\\
    0 & \quad \text{if $a_i\neq=b_{i-1}$, or $s_i=1$ }
  \end{array} \right.\]
hence the average cost constraint is
\bes
{\mathbb E}\{c(a_i,b_{i-1})\}
=P_{A_i,B_{i-1}}(0,0)+P_{A_i,B_{i-1}}(1,1)=P_{S_i}(0).
\ees
\par Note that $c(a_i,b_{i-1})$ is not required to be binary and can be
easily upgraded to more complex forms. We know that for the $BSSC(\alpha_1,\beta_1)$ \cite{asnani13}
feedback does not increase the capacity.
The definition of the constrained capacity without feedback is defined by
\bea
C_{fb}(k)=\lim_{n\rar\infty}\max_{{P}_{X^n}:\sum_{i=0}^{n}
\frac{1}{n+1}{\mathbb E}\{\sum_{i=0}^{n}{c}_{0,i}(x_i,y_{i-1})\}=\kappa}\nonumber\\
\frac{1}{n+1}I(X^n\rar Y^n)
\eea

\begin{proposition}
The capacity of the $BSSC(\alpha_1,\beta_1)$, with or without feedback, subject to the average cost
constrain ${\mathbb E}\{c(a_i,b_{i-1})\}=k$, where $\kappa=constant$, given by
\begin{IEEEeqnarray}{l}
C(\kappa)=H({\alpha_1}\kappa{+}(1{-}{\beta_1})(1{-}\kappa)){-}\kappa H({\alpha_1}){-}(1{-}\kappa)H({\beta_1})\IEEEeqnarraynumspace
\end{IEEEeqnarray}
The optimal input distribution without feedback is given by
\vspace{-0.2cm}
\bea
P^{*}_{A_i|A_{i-1}}(a_i|a_{i-1}) = \bbordermatrix{~ \cr
                  & \dfrac{1-\kappa-\gamma}{1-2\gamma} & \dfrac{\kappa-\gamma}{1-2\gamma}   \cr
                  & \dfrac{\kappa-\gamma}{1-2\gamma}   & \dfrac{1-\kappa-\gamma}{1-2\gamma} \cr},
                  \nonumber
\eea
where $\gamma={\alpha_1}{\kappa}+{\beta_1}({1-\kappa})$.
\end{proposition}
\textit{Proof:} see \cite{cbssc}.
\subsection{Symbol-By-Symbol Joint Source Channel Matching}
Recall that symbol-by-symbol joint source channel matching is achievable
if ${ R}^{na}(D)=C(\kappa)$ and if there exists an encoder decoder scheme
for $d\geq D$, such that
\begin{align}
\mathbb{P}\Big\{\sum_{i=0}^{n}{\rho}_{0,i}(T^i{X^n},T^i{Y^n})>(n+1)d\Big\}\leq\epsilon \label{edp}
\end{align}
By setting $\kappa=m$, $\alpha_1=\alpha$, $\beta_1=\beta$, then
$\frac{1-\kappa-\gamma}{1-2\gamma}=p$,
\beae
C(\kappa)&=&H({\beta}_{1}(1$-$\kappa)$+$(1$-${\alpha}_{1})\kappa)$-$\kappa{H({\alpha}_1)}$-$(1$-$\kappa)H(\beta_1)\nonumber\\
&=&H({\beta}(1$-$m)$+$(1$-${\alpha})m)$-$m{H({\alpha})}$-$(1$-$m)H(\beta)\nonumber\\
&=&H(p)$-$m{H({\alpha})}$-$(1$-$m)H(\beta)=R^{na}(D)\nonumber
\eeae
Moreover, the optimal input distribution is given by
\bea
P^{*}_{A_i|A_{i-1}}(a_i|a_{i-1}) = \bbordermatrix{~\cr
                  & p & 1-p   \cr
                  & 1-p   & p \cr},
                  \label{oidis}
\eea
\par Since the optimal input distribution is identical to the probability
distribution of the source, then no encoder is required.
Next, we check whether the average distortion is satisfied in the absence of a decoder.
The average distortion between the source symbols and the reproduction symbols, $\Delta$, is equal to
\beae
\Delta&=&{\mathbb E}[d(X_i,Y_i)]\nonumber\\
&=&{\mathbb E}[d(A_i,B_i)]\nonumber\\
&=&\sum_{A_i,B_i,B_{i-1}}d(A_i,B_i)P_{B_i|A_i,B_{i-1}}(b_i|a_i,b_{i-1})\nonumber\\
&&P_{A_i|B_{i-1}}(a_i|b_{i-1})P_{B_{i-1}}(b_{i-1})\nonumber\\
&=&(1-\beta)(1-m)+(1-\alpha)m=D\nonumber
\eeae
\par Thus, we established source channel matching of
a $BSMS(p)$ with Hamming fidelity constraint over a $BSSC(\alpha_1,\beta_1)$
subject to cost constraint, in the spirit of \cite{gastpar2003}.
A  realization of the described scheme is illustrated in Fig.~\ref{figmarex1}, where it is shown that
as the number of channel uses $n$ is increased, the single letter distortion between the source symbol sequence and
the reproduction sequence converges to the average distortion $D$.

\begin{figure}
\begin{center}
\includegraphics[bb= -10 33 400 310,scale=0.52]{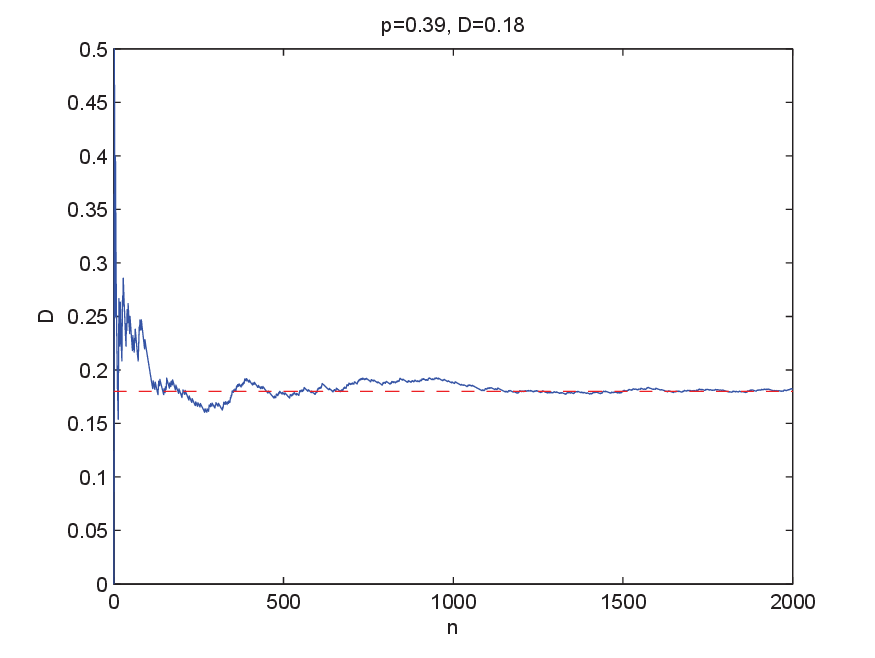}
\caption{The distortion between the source and reproduction symbols for a random realization of the source, as a function of $n$
using the optimal reproduction distribution as the channel and uncoded transmission.}
 \label{figmarex1}
\end{center}
\end{figure}

%

Next, we bound the excess distortion probability of Theorem \ref{ach}, by applying an extension of Hoeffding's inequality for MCs \cite{glynn2002}, to the
Markov process $\{Z_i \tri (Y_i,X_i): \forall i \in {\mathbb N}  \}$ (this is easily shown to hold). Set ${\rho}(x,y)=x\oplus y$ and let $S_n \tri \sum_{i=0}^{n}{\rho}(X_i,Y_i)$.
Let $d\tri \delta+ \frac{{\mathbb E}[S_n]}{n+1}, \delta >0$.
By Hoeffding's inequality \cite{glynn2002}, the excess distortion probability  is bounded by
\bea
P\Big\{S_n   > (n+1) d \Big\}\leq \exp\Big(-\frac{{\lambda}^2
((n+1)\delta -2\|f\|m/{\lambda})^2}{2(n+1){\|f\|}^2m^2}\Big)\nonumber
\eea
where ${\|f\|}\tri\sup\{y_i:i=0,1,\dots\}=1$, $m=1$,
$\lambda=\min\{p,1-p\}\min\{\alpha,\beta,1-\alpha,1-\beta\}$, for $n>2{\|f\|}m/(\lambda\delta)$.
This bound is illustrated in Fig.~\ref{figmarex2}.  Although, this bound is not tight and holds for
$n$ large enough, it shows the achievability of
Markov sources via uncoded transmission. It might be possible to compute the excess distortion probability in closed form to get tighter bounds.
\begin{figure}
\begin{center}
\centering
\includegraphics[bb= -10 33 400 310,scale=0.52]{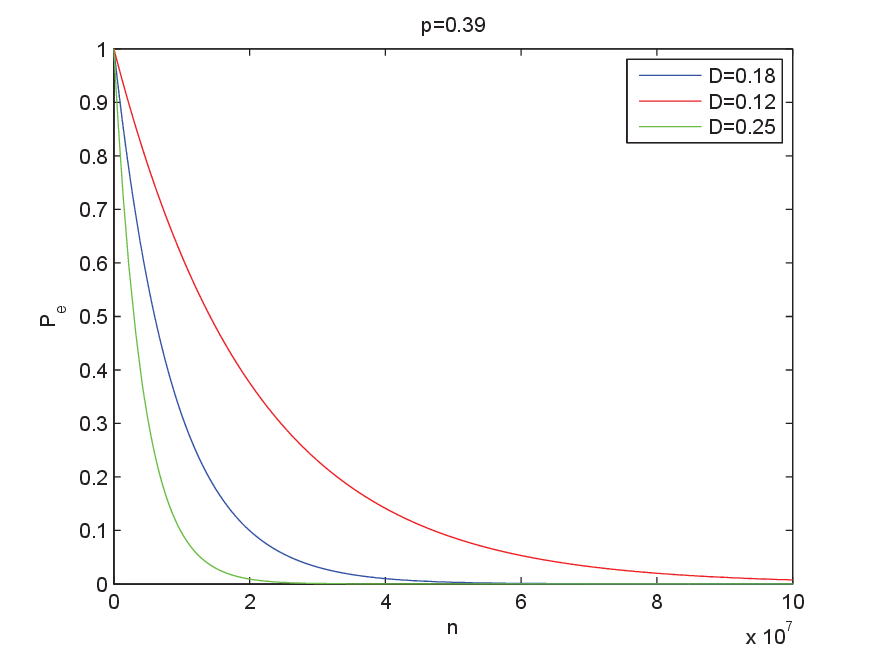}
\caption{Excess Probability of Distortion for $\delta=0.01$.}
 \label{figmarex2}
\end{center}
\end{figure}

\section{Conclusions}
This paper discusses  General Source-Channel Matching for symbol-by-symbol.
Using the nonanticipative RDF it is shows achievability of a symbol-by-symbol
code with respect to average and excess distortion probability. Then it considers
the $BSMS(p)$, it computes the nonanticipative RDF with respect
to Hamming distortion, and shows that is is matched, uncoded, over a
$BSSC(\alpha_1,\beta_1)$ subject to cost constraint but without feedback.



%



\label{Bibliography}
\bibliographystyle{IEEEtran}
\bibliography{Bibliography}

\begin{thebibliography}{10}
\providecommand{\url}[1]{#1}
\csname url@samestyle\endcsname
\providecommand{\newblock}{\relax}
\providecommand{\bibinfo}[2]{#2}
\providecommand{\BIBentrySTDinterwordspacing}{\spaceskip=0pt\relax}
\providecommand{\BIBentryALTinterwordstretchfactor}{4}
\providecommand{\BIBentryALTinterwordspacing}{\spaceskip=\fontdimen2\font plus
\BIBentryALTinterwordstretchfactor\fontdimen3\font minus
  \fontdimen4\font\relax}
\providecommand{\BIBforeignlanguage}[2]{{%
\expandafter\ifx\csname l@#1\endcsname\relax
\typeout{** WARNING: IEEEtran.bst: No hyphenation pattern has been}%
\typeout{** loaded for the language `#1'. Using the pattern for}%
\typeout{** the default language instead.}%
\else
\language=\csname l@#1\endcsname
\fi
#2}}
\providecommand{\BIBdecl}{\relax}
\BIBdecl

\bibitem{gastpar2003}
M.~Gastpar, B.~Rimoldi, and M.~Vetterli, ``To code, or not to code: lossy
  source-channel communication revisited,'' \emph{IEEE Transactions on
  Information Theory,}, vol.~49, no.~5, pp. 1147--1158, May 2003.

\bibitem{kverdu}
V.~Kostina and S.~Verdu, ``Fixed-length lossy compression in the finite
  blocklength regime: Discrete memoryless sources,'' in \emph{2011 IEEE
  International Symposium on Information Theory}, 2011, pp. 41--45.

\bibitem{berger}
T.~Berger, \emph{Rate Distortion Theory:~A Mathematical Basis for Data
  Compression}.\hskip 1em plus 0.5em minus 0.4em\relax Englewood Cliffs, NJ:
  Prentice-Hall, 1971.

\bibitem{pinsker1973}
A.~K. Gorbunov and M.~S. Pinsker, ``Nonanticipatory and prognostic epsilon
  entropies and message generation rates,'' \emph{Probl. Peredachi Inf.},
  vol.~9, no.~3, pp. 12--21, 1973,~(English version).

\bibitem{pscc}
P.~Stavrou and C.~Charalambous, ``Variational equalities of directed
  information and applications,'' in \emph{Information Theory Proceedings
  (ISIT), 2013 IEEE International Symposium on}, 2013, pp. 2577--2581.

\bibitem{charalambous-stavrou-ahmed2013}
\BIBentryALTinterwordspacing
C.~D. Charalambous, P.~A. Stavrou, and N.~U. Ahmed, ``Nonanticipative rate
  distortion function and relations to filtering theory,'' \emph{accepted to
  IEEE Transactions on Automatic Control}, 2013. [Online]. Available:
  \url{http://arxiv.org/abs/1210.1266v2}
\BIBentrySTDinterwordspacing

\bibitem{cover-pombra1989}
T.~M. Cover and S.~Pombra, ``Gaussian feedback capacity,'' \emph{IEEE
  Transactions on Information Theory}, vol.~35, no.~1, pp. 37--43, 1989.

\bibitem{charalambous-kourtellaris-hadjicostis}
C.~Charalambous, C.~Kourtellaris, and C.~Hadjicostis, ``Optimal encoder and
  control strategies in stochastic control subject to rate constraints for
  channels with memory and feedback,'' in \emph{CDC-ECC, 2011 50th IEEE
  Conference on}, dec. 2011, pp. 4522 --4527.

\bibitem{asnani13}
H.~Asnani, H.~Permuter, and T.~Weissman, ``Capacity of a post channel with and
  without feedback,'' in \emph{Information Theory Proceedings (ISIT), 2013 IEEE
  International Symposium on}, 2013, pp. 2538--2542.

\bibitem{cbssc}
\BIBentryALTinterwordspacing
C.~Kourtellaris and C.~D. Charalambous, ``Capacity of the binary state
  symmetric channel with cost constraint,'' \emph{arxiv.org}, 2014. [Online].
  Available: \url{arxiv.org}
\BIBentrySTDinterwordspacing

\bibitem{glynn2002}
P.~Glyn, W. and D.~Ormoneit, ``Hoeffding's inequality for uniform ergosic
  markov chains,'' \emph{Statistic \& Probability Letters}, vol.~56, pp. 143 --
  146, 2002.

\end{thebibliography}
%
%
%
%
%
%
%

\end{document}